\documentclass[preprint,showpacs,amsmath]{revtex4}
\usepackage{amssymb}
\usepackage{graphicx,epsfig,subfigure,dsfont,amssymb,amsthm,amsfonts,amsbsy,mathrsfs,amscd}
\def\qed{\leavevmode\unskip\penalty9999 \hbox{}\nobreak\hfill
     \quad\hbox{\leavevmode  \hbox to.77778em{%
               \hfil\vrule   \vbox to.675em%
               {\hrule width.6em\vfil\hrule}\vrule\hfil}}
     \par\vskip3pt}

\newtheorem{theorem}{Theorem}
\newtheorem{corollary}{Corollary}
\newtheorem{lemma}{Lemma}
\newtheorem{definition}{Definition}
\input amssym.def
\begin{document}

\title{The $l_1$ Norm of Coherence of Assistance}

\smallskip
\author{Ming-Jing Zhao$^1$}
\author{Teng Ma$^{2,3}$ }
\author{Quan Quan$^{4}$ }
\author{Heng Fan$^5$ }
\author{Rajesh Pereira$^6$ }

\affiliation{
$^1$School of Science,
Beijing Information Science and Technology University, Beijing, 100192, China\\
$^2$ Shenzhen Institute for Quantum Science and Engineering and Department of Physics,
South University of Science and Technology of China, Shenzhen 518055,
China\\
$^3$Shenzhen Key Laboratory of Quantum Science and Engineering, Shenzhen 518055, China\\
$^4$State Key Laboratory of Low-Dimensional Quantum Physics and Department of Physics, Tsinghua University,
Beijing, 100084, China\\
$^5$Institute of Physics, Chinese Academy of Sciences, Beijing 100190, China\\
$^6$ Department of Mathematics and Statistics, University of Guelph, N1G2W1, Canada\\
}


\begin{abstract}
We introduce and study the $l_1$ norm of coherence of assistance both theoretically and operationally.
We first provide an upper bound for the $l_1$ norm of coherence of assistance and show a necessary and sufficient condition for the saturation of the upper bound. For two and three dimensional quantum states, the analytical expression of the $l_1$ norm of coherence of assistance is given. Operationally, the mixed quantum coherence can always be increased with the help of another party's local measurement and one way classical communication since the $l_1$ norm of coherence of assistance, as well as the relative entropy of coherence of assistance, is shown to be strictly larger than the original coherence. The relation between the $l_1$ norm of coherence of assistance and entanglement is revealed. Finally, a comparison between the $l_1$ norm of coherence of assistance and the relative entropy of coherence of assistance is made.
\end{abstract}

\maketitle

\section{Introduction}

Quantum coherence is an important feature in quantum physics and is of practical significance in quantum computation and quantum communication \cite{A. Streltsov-rev,E. Chitambar-review,M. Hu,M. Hillery}.
The formulation of the resource theory of coherence was initiated in Ref. \cite{T. Baumgratz}, in which some intuitive and
computable measures of coherence are identified, for example, the $l_1$ norm of coherence and the relative entropy of coherence. These coherence measures quantify coherence by using the minimal distance between the quantum state and the set of incoherent states. Intrinsic randomness of coherence \cite{X. Yuan} and coherence concurrence \cite{X. Qi} are coherence measures defined using the convex roof construction.
Robustness of coherence is a coherence monotone and quantifies the minimal mixing required to make the state incoherent \cite{C. Napoli}, of which the witness observable has been demonstrated \cite{W. Zheng}.

Since quantum coherence is a quantum resource, the interconversion between coherent states is a major focus study. Coherence distillation is a
transformation process which transforms a
quantum state to the maximally coherent state by incoherent operations.
The reverse transformation from the maximally coherent state to a general quantum state is known as coherence cost \cite{A. Winter}.
The distillable coherence has been proven to be equal to the relative entropy of coherence and the coherence cost has been proven to be equal to the coherence formation \cite{A. Winter}.
Since the coherence distillation can not be accomplished with certainty, the framework of probabilistic
coherence distillation characterizing the relation between the
maximal success probability and the fidelity of distillation in
the one-shot setting has been developed \cite{K. Fang}.

Similar to standard coherence distillation, assisted coherence distillation is another process aided by another party typically holding a purifying quantum state and performing local measurement and one way classical communication \cite{E. Chitambar}.
A mathematical framework for the characterization of the assisted coherence distillation is then proposed in Ref. \cite{B. Regula}. These characterizations imply that the best achievable rate of assisted coherence distillation is the same no matter which class of incoherent operations the assistant performs \cite{B. Regula}.
The assisted coherence distillation has been generalized to more general
settings where two parties both can perform local measurements and communicate
their outcomes via a classical channel \cite{A. Streltsov-2017} and the assisted coherence distillation of some mixed states has been discussed in Ref. \cite{X. L. Wang}.

In the context of assisted coherence distillation, the quantity called coherence of assistance is introduced to characterize the distillation rate assisted by another party \cite{E. Chitambar}, which is defined by the maximal average relative entropy of coherence and we it denote as the relative entropy of coherence of assistance throughout our paper to avoid confusion. Analogously, we introduce the maximal average $l_1$ norm of coherence and denote it as the $l_1$ norm of coherence of assistance.

One reason we study the coherence of assistance based on the $l_1$ norm of coherence is that the $l_1$ norm of coherence is usually easy to evaluate and algebraically manipulate for a given quantum state. Any continuous weak coherence monotone which is a symmetric function of nonzero off-diagonal entries of the state must be a nondecreasing function of the $l_1$ norm of coherence \cite{H. Zhu}. Furthermore, the $l_1$ norm of coherence is an important link between different coherence measures and entanglement. For example, the $l_1$ norm of coherence is equal to the robustness of coherence for qubit states and acts as an upper bound for the robustness of coherence in high dimensional system \cite{C. Napoli}. The logarithmic $l_1$ norm of coherence is an upper bound for the relative entropy of coherence \cite{S. Rana}. Additionally the $l_1$ norm of coherence is the maximum entanglement created by incoherent operations acting on the system and an incoherent ancilla \cite{H. Zhu}. Motivated by the usefulness of the $l_1$ norm of coherence, we aim to give a characterization for the corresponding coherence of assistance.

In this paper, we study the $l_1$ norm of coherence of assistance theoretically and operationally. We first provide an upper bound of the $l_1$ norm of coherence of assistance in terms of the function of diagonal entries of quantum state. The necessary and sufficient condition when the $l_1$ norm of coherence of assistance attains this upper bound is shown. In the special case of two and three dimensional quantum states, the $l_1$ norm of coherence of assistance always achieves its upper bound.
Then we show that the $l_1$ norm of coherence of assistance as well as the relative entropy of coherence of assistance is strictly larger than
the original coherence for mixed states.
The relation between the $l_1$ norm of coherence of assistance and entanglement is revealed. Finally, a comparison between the $l_1$ norm of coherence of assistance and the relative entropy of coherence of assistance is made.

\section{The $l_1$ Norm of Coherence of Assistance}

Under a fixed reference basis $\{|i\rangle\}$,
a quantum state $\rho$ is said to be incoherent if the state is diagonal in this basis, i.e. $\rho=\sum_{i} \rho_{i}|i\rangle \langle i|$. Otherwise the quantum state is coherent. For coherent states, the $l_1$ norm of coherence and the relative entropy of coherence are two commonly used coherence measures \cite{T. Baumgratz}.
The $l_1$ norm of coherence of the quantum state $\rho=\sum_{i,j} \rho_{ij}|i\rangle \langle j|$ is the sum of the magnitudes of all the off diagonal entries
\begin{equation}
C_{l_1}(\rho)=\sum_{i\neq j} |\rho_{ij}|.
\end{equation}
The relative entropy of coherence is the difference of  von Neumann entropy between the density matrix and the diagonal matrix given by its diagonal entries,
\begin{equation}
C_r(\rho)=S(\Delta(\rho))-S(\rho),
\end{equation}
where $\Delta(\rho)$ denotes the state given by the diagonal entries of $\rho$, $S(\rho)$ is the von Neumann entropy.
Later, the relative entropy of coherence of assistance is introduced as the maximal average relative entropy of coherence
\begin{eqnarray}
C_a^{r}(\rho)=\max \sum_k p_k C_r(|\psi_k\rangle),
\end{eqnarray}
where the maximization is taken over all pure state decompositions of $\rho=\sum_k p_k |\psi_k\rangle\langle\psi_k|$ \cite{E. Chitambar}.

Inspired by the relative entropy of coherence of assistance, we introduce the
$l_1$ norm of coherence of assistance as the maximal average $l_1$ norm of coherence.
\begin{definition}
For any quantum state $\rho$, its $l_1$ norm of coherence of assistance is defined as
\begin{equation}
C_a^{l_1}(\rho)=\max \sum_k p_k C_{l_1}(|\psi_k\rangle),
\end{equation}
where the maximization is taken over all pure state decompositions of $\rho=\sum_k p_k |\psi_k\rangle\langle\psi_k|$.
\end{definition}

Analogous to the relative entropy of coherence of assistance, the $l_1$ norm of coherence of assistance $C_a^{l_1}$ has an operational interpretation.
Suppose Alice holds a state $\rho^A$ with the $l_1$ norm of coherence $C_{l_1}(\rho^A)$. Bob holds another part of the purified state of $\rho^A$. With the help of Bob performing local measurements and informing Alice of his measurement outcomes using classical communication, Alice's quantum state will be in one pure state ensemble $\{ p_k,\  |\psi_k\rangle\}$ with the $l_1$ norm of coherence $\sum_k p_k C_{l_1}(|\psi_k\rangle)$. The $l_1$ norm of coherence of Alice's state is then increased from $C_{l_1}(\rho^A)$ to $\sum_k p_k C_{l_1}(|\psi_k\rangle)$ since the $l_1$ norm of coherence is a convex function. Maximally, the $l_1$ norm of coherence can be increased to $C_a^{l_1}(\rho^A)$ in this process.

Due to the use of optimization in the definition, both the $l_1$ norm of coherence of assistance and the relative entropy of coherence of assistance are usually difficult to calculate.
However, for the relative entropy of coherence of assistance, it is bounded from above by $S(\Delta(\rho))$ \cite{E. Chitambar},
$C_a^{r}(\rho)\leq S(\Delta(\rho))$,
and equality holds if and only if there is a pure state decomposition $\{p_k,\  |\psi_k\rangle\}$ of $\rho$ such that $\Delta(|\psi_k\rangle\langle\psi_k|)=\Delta(\rho)$ for all $k$ \cite{M. J. Zhao}. For the $l_1$ norm of coherence of assistance, it is obvious that $0\leq C_a^{l_1}(\rho)\leq n-1$ for $n$ dimensional quantum state $\rho$. $C_a^{l_1}(\rho)=0$ if and only if $\rho$ is incoherent pure state. $C_a^{l_1}(\rho)\leq n-1$ because the maximum of the $l_1$ norm of coherence of $n$ dimensional quantum state is $n-1$.
Now we show an analytical upper bound for the $l_1$ norm of coherence of assistance.

\begin{theorem}\label{th ca upper bound}
The $l_1$ norm of coherence of assistance of $\rho=\sum \rho_{ij}|i\rangle \langle j|$ is restricted by its diagonal entries as
\begin{equation}\label{eq upper bound}
C_a^{l_1}(\rho)\leq \sum_{i\neq j}\sqrt{\rho_{ii}\rho_{jj}}.
\end{equation}
Equality holds if and only if there exist a pure state decomposition $\{p_k,\  |\psi_k\rangle\}$ of $\rho$ such that $\Delta(|\psi_k\rangle\langle\psi_k|)=\Delta(\rho)$ for all $k$.
\end{theorem}

\begin{proof}
Suppose $\{p_k,\  |\psi_k\rangle\}$ is an optimal decomposition for $\rho$ such that $C_a^{l_1}(\rho)= \sum_k p_k C_{l_1}(|\psi_k\rangle)$. Since $|\psi_k\rangle\langle\psi_k|$ is a pure state and rank one, then we assume $|\psi_k\rangle\langle\psi_k|=\sum_{i,j} \sqrt{a_{ii}^{(k)}a_{jj}^{(k)}}e^{i\theta_{ij}^{(k)}}|i\rangle\langle j|$ with $\theta_{ij}^{(k)}=-\theta_{ji}^{(k)}$ and positive $a_{ii}^{(k)}$ for all $i$, $j$ and $k$. The $l_1$ norm of coherence of $|\psi_k\rangle\langle\psi_k|$ is $\sum_{i\neq j} \sqrt{a_{ii}^{(k)}a_{jj}^{(k)}}$,
which gives rise to the $l_1$ norm of coherence of assistance of $\rho$ as
\begin{equation}\label{eq upper bound hold eq}
\begin{array}{rcl}
C_a^{l_1}(\rho)&=&\sum_{i\neq j} \sum_k p_k  \sqrt{a_{ii}^{(k)}a_{jj}^{(k)}}\\
&\leq& \sum_{i\neq j} \sqrt{\sum_k p_k a_{ii}^{(k)}} \sqrt{\sum_k p_k a_{jj}^{(k)}}\\
&=& \sum_{i\neq j} \sqrt{\rho_{ii} \rho_{jj}},
\end{array}
\end{equation}
where we have utilized the Cauchy-Schwarz inequality $(\sum_k a_k b_k)^2 \leq {\sum_k a_k^2 \sum_k b_k^2}$ in the above inequality, and $\rho_{ii}=\sum_k p_k a_{ii}^{(k)}$ for the last equation.

For Eq. (\ref{eq upper bound hold eq}), the left hand side equals to the right hand side if and only if vectors $\vec{v}_k=(a_{11}^{(k)}, a_{22}^{(k)},\cdots,a_{nn}^{(k)})$ are parallel for all $k$. This implies that the diagonal entries of $|\psi_k\rangle\langle\psi_k|$ are all the same and are equal to that of quantum state $\rho$ for all $k$. Hence, the $l_1$ norm of coherence of assistance reaches its upper bound in Eq. (\ref{eq upper bound}) if and only if there exist a pure state decomposition $\{p_k,\  |\psi_k\rangle\}$ of $\rho$ such that $\Delta(|\psi_k\rangle\langle\psi_k|)=\Delta(\rho)$ for all $k$.
\end{proof}

Notice that if $\rho$ is an $n$ dimensional state, then $\sum_{i\neq j} \sqrt{\rho_{ii} \rho_{jj}}\leq n-1$ with equality if and only if all diagonal entries of $\rho$ are equal to $1/n$. So the upper bound for Theorem \ref{th ca upper bound} is always an improvement over the $n-1$ bound.
Theorem \ref{th ca upper bound} not only provides an upper bound for the $l_1$ norm of coherence of assistance, but also gives a necessary and sufficient condition for the pure state decomposition when the $l_1$ norm of coherence of assistance reaches its upper bound in Eq. (\ref{eq upper bound}). Coincidentally, this condition is the same as the condition that the relative entropy of coherence of assistance $C_a^{r}(\rho)$ reaches its upper bound $S(\Delta(\rho))$ \cite{E. Chitambar, M. J. Zhao}.
In other words,
$C_a^{l_1}(\rho)= \sum_{i\neq j}\sqrt{\rho_{ii}\rho_{jj}} \Longleftrightarrow C_a^{r}(\rho)=S(\Delta(\rho))$.
In Ref. \cite{B. Regula}, it shows the relative entropy of coherence of assistance $C_a^{r}(\rho)$ is equal to its upper bound $S(\Delta(\rho))$ for all two and three dimensional quantum states. That is, all quantum states $\rho$ in two and three dimensional systems have a pure state decomposition $\{p_k,\ |\psi_k\rangle\}$ such that the density matrix of each pure state $|\psi_k\rangle\langle\psi_k|$ has the same diagonal entries as the original density matrix, $\Delta(|\psi_k\rangle\langle \psi_k|)=\Delta(\rho)$ for all $k$. Such a decomposition is an optimal decomposition reaching the upper bound in Theorem \ref{th ca upper bound}. Therefore, the $l_1$ norm of coherence of assistance is equal to its upper bound in Eq. (\ref{eq upper bound}) for all two and three dimensional quantum states. Unfortunately, this result is not true for four dimensional system as one counterexample has been pointed out in Ref. \cite{E. Chitambar}.

\begin{corollary}\label{th qubit ca}
In two and three dimensional systems, the $l_1$ norm of coherence of assistance of state $\rho=\sum_{i,j} \rho_{ij}|i\rangle \langle j|$ is
$C_a^{l_1}(\rho)=\sum_{i\neq j}\sqrt{\rho_{ii}\rho_{jj}}$.
\end{corollary}

For two dimensional systems, we can give an optimal decomposition for all quantum states using the following lemma \cite{M. J. Zhao}.

\begin{lemma}\label{lemma different magnitude}
For any $2\times 2$ Hermitian matrix $
A=\left(
\begin{array}{ccccccc}
a_{11} & a_{12}\\
a_{12}^* & a_{22}
\end{array}
\right)$, it has a decomposition as
\begin{equation}
A= p_0  |\psi_0\rangle\langle\psi_0|+p_1  |\psi_1\rangle\langle\psi_1|,
\end{equation}
where
\begin{equation}\label{eq in lemma different magnitude}
\begin{array}{rcl}
|\psi_0\rangle&=&\sqrt{a_{11}}|1\rangle + \sqrt{a_{22}}e^{-{\rm i}\arg(a_{12})}|2\rangle,\\
|\psi_1\rangle&=&\sqrt{a_{11}}|1\rangle - \sqrt{a_{22}}e^{-{\rm i}\arg(a_{12})}|2\rangle,
\end{array}
\end{equation}
and $p_0=\frac{1}{2}(1+|a_{12}|/\sqrt{a_{11}a_{22}})$, $p_1=\frac{1}{2}(1-|a_{12}|/\sqrt{a_{11}a_{22}})$ for nonzero $a_{11}$ and $a_{22}$, $\arg(a_{12})$ is the argument of $a_{12}$.
\end{lemma}

This lemma gives a pure state decomposition for all $2\times 2$ Hermitian matrices including $2\times 2$ density matrices. Such a decomposition is an optimal decomposition for both the $l_1$ norm of coherence of assistance and the relative entropy of coherence of assistance.

\section{A Strict $l_1$ Norm of Coherence Inequality for Mixed States}

Coherence distillation is one process that extracts pure maximal coherence from a mixed state by incoherent operations \cite{A. Winter}. This process gives an operational way to obtain maximal coherence as a resource, which requires many copies of quantum state experimentally and gets the maximally coherent state with some probability. It is proved that there is no bound coherence and all coherent states are distillable \cite{A. Winter}. Another operational way to obtain the coherence resource is the assisted coherence distillation
with the assistance of another party's local measurement and one way classical communication.
Here we shall prove that all mixed coherence can be increased in this scenario.

\begin{lemma}\label{lemma different decomposition}
Suppose $\rho=\sum_{l=1}^n \lambda_l |\psi_l\rangle \langle \psi_l|$ and $\rho=\sum_{k=1}^n p_k |\phi_k\rangle \langle \phi_k|$ are two arbitrary pure state decompositions of given quantum state $\rho$ with $\sum_{l=1}^n \lambda_l=\sum_{k=1}^n p_k=1$, $0\leq \lambda_l\leq 1$, $0\leq p_k\leq 1$ for $l,k=1,\cdots,n$.
Then these two pure decompositions are related by a unitary transformation:
\begin{equation}\label{eq relation between two decomposition}
\sqrt{p_k}|\phi_k\rangle=\sum_{l=1}^n U_{kl} \sqrt{\lambda_l}|\psi_l\rangle,\ \ \ k=1,\cdots,n,
\end{equation}
where $U=(U_{kl})$ is a unitary transformation \cite{E. Sch}.
\end{lemma}

In fact the normalization condition for $\rho$ in this Lemma is not necessary and $\{|\psi_i\rangle\}$ and $\{|\phi_k\rangle\}$ are not necessary to be normalized.

\begin{theorem}\label{th increase coh}
For arbitrary mixed quantum state $\rho=\sum_{i,j=1}^n \rho_{ij}|i\rangle \langle j|$, it has
\begin{equation}\label{eq strict ineq}
C_a^{l_1}(\rho)>C_{l_1}(\rho).
\end{equation}
\end{theorem}

\begin{proof}
Note that the $l_1$ norm of coherence is a convex function, $\sum_i p_i C_{l_1}(\rho_i)\geq C_{l_1}(\sum_i p_i\rho_i)$, hence it is obvious that $C_a^{l_1}(\rho)\geq C_{l_1}(\rho)$. Next we prove that this inequality is strict.

Suppose $\rho=\sum_{s=1}^n \lambda_s |\psi_s\rangle \langle \psi_s|$ is the spectral decomposition with eigenstate $|\psi_s\rangle=\sum_{j=1}^n a_j^{(s)} |j\rangle$, $\sum_{s=1}^n \lambda_s=1$, $0\leq \lambda_s\leq 1$, $s=1,\cdots, n$.
Since $\rho$ is mixed, considering its two nonzero eigenvalues and corresponding eigenstates, one can pick out two two dimensional subvectors from these eigenstates respectively such that these subvectors are not parallel. Without loss of generality, we assume
(1) $\lambda_1$ and $\lambda_2$ are nonzero.  (2) two unnormalized two dimensional vectors composed by the first two components of $|{\psi}_1\rangle$ and $|{\psi}_2\rangle$ denoted by $|\tilde{\psi}_1\rangle=(a_1^{(1)},a_2^{(1)})^T$ and $|\tilde{\psi}_2\rangle=(a_1^{(2)},a_2^{(2)})^T$ are not parallel, where superscript $T$ means transposition.

The first case, if $a_1^{(1)*}a_2^{(1)}$ and $a_1^{(2)*}a_2^{(2)}$ have different arguments, then
\begin{equation}\label{eq proof different magnitude}
\begin{array}{rcl}
C_a^{l_1}(\rho)&\geq& \sum_{i=1}^n \lambda_s C_{l_1}(|\psi_s\rangle)\\
&=&  \sum_{i=1}^n \sum_{i\neq j} \lambda_s  |a_i^{(s)*}a_j^{(s)}|\\
&>&   \sum_{i\neq j}  |\sum_{s=1}^n  \lambda_s a_i^{(s)*}a_j^{(s)}|\\
&=&  \sum_{i\neq j}  |\rho_{ij}|\\
&=&C_{l_1}(\rho).
\end{array}
\end{equation}

The second case, if $a_1^{(1)*}a_2^{(1)}$ and $a_1^{(2)*}a_2^{(2)}$ have the same arguments.
Let
\begin{equation}
U=\sum_{s,t=1,2} U_{st} |\psi_s\rangle \langle \psi_t|+\sum_{s=3}^n |\psi_s\rangle \langle \psi_s|,
\end{equation}
be a unitary transformation and
\begin{equation}\label{eq two dim U}
\tilde{U}=\left(
\begin{array}{ccccccc}
U_{11} & U_{12}\\
U_{21} & U_{22}
\end{array}
\right)
\end{equation}
be a two dimensional unitary transformation induced by $U$. Under the transformation of $U$, we can get another pure state decomposition of $\rho$, $\rho=\sum_{k=1}^n p_k |\phi_k\rangle \langle \phi_k|$, with pure state $|\phi_k\rangle=\sum_{i=1}^n b_i^{(k)} |i\rangle$, $\sum_{k=1}^n p_k=1$, $0\leq p_k\leq 1$, $k=1,\cdots,n$, which is related with $\{|\psi_s\rangle\}$ according to Eq. (\ref{eq relation between two decomposition}).
Denote
\begin{equation}
|\tilde{\phi}_k\rangle=(b_1^{(k)},b_2^{(k)})^T
\end{equation}
as the unnormalized vector composed of the first two entries of $|{\phi}_k\rangle$ with $k=1,2$. Then $\{p_k,\  |\tilde{\phi}_k\rangle\}_{k=1,2}$ and $\{\lambda_s,\  |\tilde{\psi}_s\rangle\}_{s=1,2}$ can be regarded as two decompositions of $2\times 2$ Hermitian matrix $\tilde{A}$ composed by the entries of the Hermitian matrix $A=\sum_{s=1,2} \lambda_s |{\psi}_s\rangle \langle {\psi}_s|=\sum_{k=1,2} p_k |{\phi}_k\rangle \langle {\phi}_k|$ in positions (1,1) (1,2), (2,1), (2,2) as
\begin{eqnarray}
\tilde{A}=\sum_{s=1,2} \lambda_s |\tilde{\psi}_s\rangle \langle \tilde{\psi}_s|=\sum_{k=1,2} p_k |\tilde{\phi}_k\rangle \langle \tilde{\phi}_k|.
\end{eqnarray}
These two decompositions $\{p_k,\  |\tilde{\phi}_k\rangle\}_{k=1,2}$ and $\{\lambda_s,\  |\tilde{\psi}_s\rangle\}_{s=1,2}$ are connected by the unitary transformation $\tilde{U}$ in Eq. (\ref{eq two dim U}) by relation Eq. (\ref{eq relation between two decomposition}).
Since unnormalized vectors $|\tilde{\psi}_1\rangle$ and $|\tilde{\psi}_2\rangle$ are different, the two dimensional Hermitian matrix $\tilde{A}$ is full rank. According to Lemma \ref{lemma different magnitude}, there exists a unitary transformation $\tilde{U}$ such that $\{p_k,\  |\tilde{\phi}_k\rangle\}_{k=1,2}$ is a decomposition in form of Eq. (\ref{eq in lemma different magnitude}) with $0<p_k<1$, $k=1,2$. This means $b_1^{(1)*}b_2^{(1)}$ and $b_1^{(2)*}b_2^{(2)}$ have different arguments. Therefore we derive one decomposition $\{p_k,\  |{\phi}_k\rangle\}_{k=1}^n$ of $\rho$ satisfying the first case. As in Eq. (\ref{eq proof different magnitude}), one gets $C_a^{l_1}(\rho)\geq \sum_{k=1}^n p_k C_{l_1}(|\phi_k\rangle)>C_{l_1}(\rho)$.
\end{proof}

Since the $l_1$ norm of coherence is nonnegative and the $l_1$ norm of coherence of assistance of mixed state is strictly larger than it,
this implies the positivity of the $l_1$ norm of coherence of assistance.
\begin{corollary}
For mixed quantum state $\rho$, the $l_1$ norm of coherence of assistance is strictly positive, $C_a^{l_1}(\rho)>0$.
\end{corollary}

Here we show one protocol of obtaining a larger $l_1$ norm of coherence with the help of another party using local measurement and one way classical communication. Now we give one example in four dimensional system. Let $\rho_A=\frac{1}{4}\sum_{i=0}^3 |i\rangle\langle i|$ be a maximally mixed state held by Alice and the initial coherence in A is zero. As a purification with another party held by Bob
we first prepare a pure entangled state $|\psi\rangle_{AB}=\frac{1}{2}\sum_{i=0}^3 |\psi_i\rangle_A|i \rangle_B$, with maximally coherent state $|\psi_i\rangle=\frac{1}{2}\sum_j (-1)^{\delta_{ij}}|j\rangle$ for $\delta_{ij}=1$ when $i=j$ and zero otherwise as $\rho_A=\frac{1}{4}\sum_{i=0}^3 |i\rangle\langle i|=\frac{1}{4}\sum_{i=0}^3 |\psi_i\rangle_{AB}\langle\psi_i|$. Then Bob performs von Neumann measurements on the basis $\{|i\rangle_B\}$.
If Bob's component is projected to state $|i\rangle_B$, the state of Alice will be collapsed to $|\psi_i\rangle_A$, with maximal $l_1$ norm of coherence, $i=0,1,2,3$.
After receiving Bob's measurement outcomes via
classical communication channel, Alice can obtain her state in a four state ensemble $\{p_i=\frac{1}{4}, |\psi_i\rangle\}$ with the $l_1$ norm of coherence $\frac{1}{4}\sum_{i=0}^3 C_{l_1}(|\psi_i\rangle_{AB}) =3$. Therefore the final $l_1$ norm of coherence for Alice is increased from $C_{l_1}(\rho_A)=0$ to the maximum $C_a^{l_1}(\rho)=3$.

\section{Relation Between the $l_1$ norm of coherence of Assistance and Entanglement}

The curious observation that the coherence of $\rho=\sum_{i,j} \rho_{ij}|i\rangle \langle j|$ is closely related to the entanglement of $\rho_{mc}=\sum_{i,j} \rho_{ij}|ii\rangle \langle jj|$ \cite{E. Rains} was first noticed in Ref. \cite{A. Winter}:
\begin{equation}
\rho=\sum_{i,j} \rho_{ij}|i\rangle \langle j|\ \longleftrightarrow \ \rho_{mc}=\sum_{i,j} \rho_{ij}|ii\rangle \langle jj|.
\end{equation}
For example, the coincidence of coherent cost and coherence of formation is identified with the coincidence of entanglement cost and entanglement of formation \cite{A. Winter}.
The relative entropy of coherence of assistance of $\rho$ is equal to the entanglement of assistance \cite{D. DiVincenzo} of $\rho_{mc}$,
$C_a^{r}(\rho)=E_a(\rho_{mc})$
with $E_a(\rho)=\max \sum_i p_i E(|\psi_i\rangle)=\max \sum_i p_i S(Tr_B(|\psi_i\rangle\langle\psi_i|))$,
where the maximization is taken over all pure state decompositions of $\rho=\sum_i p_i |\psi_i\rangle\langle\psi_i|$ \cite{E. Chitambar}. Similarly, we shall show the $l_1$ norm of coherence of assistance corresponds to the entanglement called the convex-roof extended negativity of assistance \cite{J. S. Kim}, which is defined as $N_a(\rho)=\max \sum_i p_i N(|\psi_i\rangle)$,
where the maximization is taken over all pure state decompositions of $\rho=\sum_i p_i |\psi_i\rangle\langle\psi_i|$, $N(\rho)$ is the negativity of $\rho$ which is a well known entanglement measure defined as the sum of all absolute values of negative eigenvalues of $\rho^{PT}$ where the superscript $PT$ denotes the partial transposition \cite{G. Vidal}.

\begin{theorem}
The $l_1$ norm of coherence of assistance of $\rho=\sum_{i,j} \rho_{ij}|i\rangle \langle j|$ is related to the convex-roof extended negativity of assistance of maximally correlated state $\rho_{mc}=\sum_{i,j} \rho_{ij}|ii\rangle \langle jj|$ as
\begin{equation}
C_a^{l_1}(\rho)=2N_a(\rho_{mc}).
\end{equation}
\end{theorem}

\begin{proof}
Note that for maximally correlated state $\rho_{mc}$, its pure state decompositions are all in the Schmidt form $|\psi^\prime\rangle=\sum_i a_i|ii\rangle$ \cite{M. J. Zhao2008}. The negativity of $|\psi^\prime\rangle=\sum_i a_i|ii\rangle$ is related to the $l_1$ norm of coherence of $|\psi\rangle=\sum_i a_i|i\rangle$ as $2N(|\psi^\prime\rangle)=C_{l_1}(|\psi\rangle)=\sum_{i\neq j} |a_i^*a_j|$ \cite{H. Zhu, S. Rana}.
Therefore, if $\{ p_k,\  |\psi^{\prime}_k\rangle\}$ is an optimal decomposition for $\rho_{mc}$ such that $N_a(\rho_{mc})=\sum_k p_k N(|\psi^{\prime}_k\rangle)$ with $|\psi^{\prime}_k\rangle=\sum_i a_i^{(k)}|ii\rangle$. Then $\{ p_k,\  |\psi_k\rangle\}$ with $|\psi_k\rangle=\sum_i a_i^{(k)}|i\rangle$ is the optimal decomposition for $\rho$ such that $C_a^{l_1}(\rho)=\sum_k p_k C_{l_1}(|\psi_k\rangle)$.
\end{proof}

\section{Comparison between the $l_1$ norm of coherence of assistance and the relative entropy of coherence of assistance}

The $l_1$ norm of coherence of assistance and the relative entropy of coherence of assistance exhibit many similarities.
Parallel to the result in Theorem \ref{th increase coh}, we can prove the relative entropy of coherence of assistance is strictly larger than the original relative entropy of coherence for all mixed state.

\begin{theorem}
For any mixed quantum state $\rho$, it has
\begin{eqnarray}\label{eq strict ineq for r}
C_a^{r}(\rho)>C_{r}(\rho).
\end{eqnarray}
\end{theorem}

\begin{proof}
Suppose $\rho=\sum_{s=1}^m \lambda_s |\psi_s\rangle \langle \psi_s|$ is the spectral decomposition with $\sum_{s=1}^m \lambda_s=1$ and $0< \lambda_s\leq 1$. Then
\begin{equation}
\begin{array}{rcl}
C_a^{r}(\rho)&\geq& \sum_{s=1}^m \lambda_s C_r(|\psi_s\rangle \langle \psi_s|)\\
&=& \sum_{s=1}^m \lambda_s S(\Delta(|\psi_s\rangle \langle \psi_s|))\\
&\geq &  S(\sum_{s=1}^m \lambda_s\Delta(|\psi_s\rangle \langle \psi_s|))-H(\lambda_s)\\
&=& S(\Delta(\rho))-S(\rho)\\
&=&C_{r}(\rho),
\end{array}
\end{equation}
where $H(\lambda_s)=-\sum_{s=1}^m \lambda_s \log \lambda_s$, the third inequality is according to the property of entropy and becomes equality if and only if $\{\Delta(|\psi_s\rangle \langle \psi_s|)\}$ have support on orthogonal subspaces \cite{M. A. Nielsen}. Since $m$ projected eigenstates $\{\Delta(|\psi_s\rangle \langle \psi_s|)\}$ are in the $m$ dimensional space spanned by $\{ |i\rangle \langle i| \}_{i=1}^m$, so they are orthogonal if and only if they are all supported on one dimensional subspaces spanned by one of the elements in $\{ |i\rangle \langle i| \}_{i=1}^m$, which implies all eigenstates $\{|\psi_s\rangle \langle \psi_s|\}$ themselves are incoherent under the reference basis. Therefore, the third inequality becomes equality if and only if $\rho$ is incoherent, which means the third inequality is strict for all mixed coherent quantum states. In this case the relative entropy of coherence of assistance $C_a^{r}(\rho)$ is strictly larger than the relative entropy of coherence $C_{r}(\rho)$ for mixed coherent state $\rho$.

For mixed incoherent states, their relative entropy of coherence are all zero, $C_{r}(\rho)=0$. But, under the unitary transformations, we can get other coherent pure state decomposition with positive average relative entropy of coherence. This means that the relative entropy of coherence of assistance $C_a^{r}(\rho)$ is strictly larger than the relative entropy of coherence $C_{r}(\rho)$ for all mixed incoherent quantum states. Therefore, inequality (\ref{eq strict ineq for r}) holds true for all mixed states.
\end{proof}

\begin{corollary}
For any mixed quantum state $\rho$, the relative entropy of coherence of assistance is strictly positive, $C_a^{r}(\rho)>0$.
\end{corollary}

The $l_1$ norm of coherence of assistance can be shown to be an upper bound for the relative entropy of coherence of assistance.

\begin{theorem}
For any quantum state $\rho$, the $l_1$ norm of coherence of assistance and the relative entropy of coherence of assistance satisfy the relation
\begin{equation}
C_a^{r}(\rho)\leq C_a^{l_1}(\rho).
\end{equation}
\end{theorem}

\begin{proof}
For any quantum state $\rho$, suppose $\{ p_k,\  |\psi_k\rangle\}$ is an optimal decomposition of $\rho$ such that $C_a^{r}(\rho)=\sum_k p_k C_r(|\psi_k\rangle)$. Since the relative entropy of coherence is bounded by the $l_1$ norm of coherence from above for all pure states \cite{S. Rana}, then $C_a^{r}(\rho)=\sum_k p_k C_r(|\psi_k\rangle)\leq \sum_k p_k C_{l_1}(|\psi_k\rangle)\leq C_a^{l_1}(\rho)$.
\end{proof}

Now we come to compare the $l_1$ norm of coherence of assistance and the relative entropy of coherence of assistance in the following table.
Although their original coherence measures are defined from different point of views, one uses the $l_1$ norm from mathematics and the other uses the entropy from information theory, their coherence of assistance exhibit many similarities including their definitions, their upper bounds with the saturation condition, their optimal pure state decompositions and their strict positivity for mixed states. The reason lies in the fact that they are both defined using the maximization as well as the property of the coherence measures.

\begin{table}[!h]
\tabcolsep 0pt
\vspace*{-12pt}
\newcommand{\tabincell}[2]{\begin{tabular}{@{}#1@{}}#2\end{tabular}}
\begin{center}
\def\temptablewidth{0.9\textwidth}
{\rule{\temptablewidth}{1pt}}
\begin{tabular*}
{\temptablewidth}{@{\extracolsep{\fill}}|c|c|c|}
& $C_a^{r}$ & $C_a^{l_1}$ \\   \hline
Definition & $C_a^{r}(\rho)=\max \sum_k p_k C_r(|\psi_k\rangle)$ &  $C_a^{l_1}(\rho)=\max \sum_k p_k C_{l_1}(|\psi_k\rangle)$ \\  \hline
Upper bound & $C_a^{r}(\rho)\leq S(\Delta(\rho))$
&  $C_a^{l_1}(\rho)\leq \sum_{i\neq j} \sqrt{\rho_{ii}\rho_{jj}}$ \\  \hline
\tabincell{c} {Analytical formula \\for qubit and qutrit state} & $C_a^{r}(\rho)= S(\Delta(\rho))$ &  $C_a^{l_1}(\rho)= \sum_{i\neq j} \sqrt{\rho_{ii}\rho_{jj}}$ \\  \hline
Range& $0\leq C_a^{r}(\rho) \leq \log n$ & $0\leq C_a^{l_1}(\rho)\leq n-1$  \\ \hline
Relation with original coherence &
 \tabincell{c} {$C_a^{r}(\rho)>C_r(\rho) \text{ if $\rho$ is mixed}$ \\
   $C_a^{r}(\rho)=C_r(\rho) \text{ if $\rho$ is pure}$ }&
   \tabincell{c} {$C_a^{l_1}(\rho)>C_{l_1}(\rho) \text{ if $\rho$ is mixed}$\\$C_a^{l_1}(\rho)=C_{l_1}(\rho) \text{ if $\rho$ is pure}$ } \\ \hline
Relation with Entanglement & $C_a^{r}(\rho)=E_a(\rho_{mc})$     & $C_a^{l_1}(\rho)=2N_a(\rho_{mc})$  \\ \hline
\end{tabular*}
 {\rule{\temptablewidth}{1pt}}
 \caption{Table of comparison between $C_a^{r}$ and $C_a^{l_1}$ for $n$ dimensional quantum state $\rho=\sum_{i,j=1}^n \rho_{ij}|i\rangle \langle j|$.
  }
 \end{center}
 \end{table}

Based on these results, we can divide all quantum states into three classes with respect to the relation between the coherence of assistance and the original coherence. Here we do not distinguish between the $l_1$ norm of coherence or the relative entropy of coherence because it does not matter. The first class $\mathcal{S}_1$ is the set of all pure incoherent states. This class of states is incoherent and it stays incoherent even under the help of another part with local measurements and one way classical communication. Such states satisfy $C_a(\rho)=C(\rho)=0$. The second class $\mathcal{S}_2$ is the set of all pure coherent states. This class of states is coherent, but its coherence can not be increased under the help of another part with local measurements and one way classical communication. Such states satisfy $C_a(\rho)=C(\rho)>0$. The third class $\mathcal{S}_3$ is the set of all mixed states. No matter whether this class of states is incoherent or coherent, their coherence can be increased under the help of another part with local measurements and one way classical communication. Such states satisfy $C_a(\rho)>C(\rho)\geq 0$.
In Ref. \cite{T. Ma}, the difference between the coherence of assistance and the original coherence is called the accessible coherence which measures the increased coherence. For quantum states in $\mathcal{S}_3$, the accessible coherence is always positive. One can gain more coherence if one knows the corresponding ensemble of the state.

\section{Conclusions and Discussions}

To summarize, we have proposed and systematically studied the $l_1$ norm of coherence of assistance. We have provided an upper bound for the $l_1$ norm of coherence of assistance in terms of the function of diagonal entries of quantum state. A necessary and sufficient condition when the $l_1$ norm of coherence of assistance reaches its upper bound has been shown. In the special case of two and three dimensional quantum states, the $l_1$ norm of coherence of assistance is always equal to this upper bound. After that, we have shown the $l_1$ norm of coherence of assistance as well as the relative entropy of coherence of assistance is strictly larger than the original coherence for mixed states. The relation between the $l_1$ norm of coherence of assistance and entanglement has been revealed. We give a comparison between the $l_1$ norm of coherence of assistance and the relative entropy of coherence of assistance at last.

Operationally the $l_1$ norm of coherence of assistance is the maximal $l_1$ norm of coherence one can gain with the help of another party's local measurement and one way classical communication. This may provide more available and useful resource for quantum information processing.
An experimental realization in linear optical systems for obtaining the maximal relative entropy of coherence for two dimensional quantum states in assisted distillation
protocol has been presented \cite{K. D. Wu}. We hope this work will promote the further study of coherence resource.

\bigskip
\noindent{\bf Acknowledgments}\, Ming-Jing Zhao thanks the Department of Mathematics and Statistics, University of Guelph, Canada for hospitality. This work is supported by the NSF of China (Grant Nos. 11401032) and the China Scholarship
Council (Grant No. 201808110022). The authors appreciate the
valuable suggestions and comments by the anonymous referee.

\end{document}